\documentclass[letterpaper,11pt]{article}
\usepackage[english]{babel}
\usepackage[T1]{fontenc}
\usepackage{latexsym,amsmath}
\usepackage{amssymb}
\usepackage[dvips]{graphicx}
\usepackage{fullpage}

{\makeatletter
 \gdef\and{\qquad}
 \gdef\maketitle{{
   \def\thefootnote{\@fnsymbol\c@footnote}
   \begin{center}
     \LARGE \@title
   \end{center}\medskip
   \centerline{\@author}\bigskip
   \@thanks
   \setcounter{footnote}{0}}}}

\newtheorem{thm}{Theorem}
\newtheorem{lemma}{Lemma}

\newenvironment{proof}{\noindent {\bf Proof:}}{{\hfill $\Box$}}

\newcommand{\toppath}{top path}
\newcommand{\bottomtree}{bottom tree}

\makeatletter
\newcommand{\listoftodo}{\cleardoublepage\chapter*{TODOs}\@starttoc{tod}}

\newcommand{\todo}[1]{}
\newcommand{\l@todo}[2]{\par\noindent\parbox[t]{2cm}{Side #2}%
                                     \parbox[t]{\linewidth-2cm}{#1}}
\makeatother

\title{An $O(\log \log n)$-Competitive Binary Search Tree with Optimal Worst-Case Access Times}

\author{Prosenjit Bose \thanks{School of Computer Science, Carleton
University. The authors are partially supported by NSERC and MRI. Email:
\{jit,karim,vida\}@cg.scs.carleton.ca.} \and Karim Dou\"ieb
\footnotemark[1] \and Vida Dujmovi\'c \footnotemark[1] \and Rolf Fagerberg
\thanks{Department of Mathematics and Computer Science, University of
Southern Denmark. Email: rolf@imada.sdu.dk.}}

\date{}

\begin{document}
\sloppy
\maketitle

\abstract{We present the \emph{zipper tree}, an $O(\log \log
n)$-competitive online binary search tree that performs each access in
$O(\log n)$ worst-case time. This shows that for binary search trees,
optimal worst-case access time and near-optimal amortized access time can
be guaranteed simultaneously. \todo{Note: took out a sentence of abstract
(seemed unclear to me). Rolf.}
}

\section{Introduction}
A \emph{dictionary} is a basic data structure for storing and retrieving
information. The \emph{binary search tree} (BST) is a well-known and widely
used dictionary implementation which combines efficiency with flexibility
and adaptability to a large number of purposes. It constitutes one of the
fundamental data structures of computer science.

In the past decades, many BST schemes have been developed which perform
element accesses (and indeed many other operations) in $O(\log n)$ time,
where $n$ is the number of elements in the tree. This is the optimal
single-operation worst-case access time in a comparison based model.
Turning to \emph{sequences} of accesses, it is easy to realize that for
specific access sequences, there may be BST algorithms which serve $m$
accesses in less than $\Theta(m \log n)$ time. A common way to evaluate how
well the performance of a given BST algorithm adapts to individual
sequences, is \emph{competitive analysis}: For an access sequence $X$,
define ${\rm OPT}(X)$ to be the minimum time needed by any BST algorithm to
serve it.  To make this precise, a more formal definition of a BST model
and of the sequences considered is needed---standard in the area is to use
the binary search tree model (BST model) defined by Wilber~\cite{wilber},
in which the only existing non-trivial lower bounds on ${\rm OPT}(X)$ have
been proven~\cite{tango,wilber}.
%
%
A given BST algorithm $A$ is then said to be $f(n)$-\emph{competitive} if
it performs $X$ in $O(f(n)\, {\rm OPT}(X))$ time for all $X$.

In 1985, Sleator and Tarjan~\cite{splay} developed a BST called \emph{splay
trees}, which they conjectured to be $O(1)$-competitive. Much of the research
on BST's efficiency on individual input sequences has grown out of this
conjecture.  However, despite decades of research, the conjecture is still
open. More generally, it is unknown if there exist asymptotically
optimal BST
data structures. In fact, for many years the best known competitive ratio
for any BST structure was $O(\log n)$, which is achieved by plain
balanced static trees.

This situation was recently improved by Demaine~\emph{et al.}, who in a
seminal paper~\cite{tango} developed a $O(\log \log n)$-competitive BST
structure, called the \emph{tango tree}. This was the first improvement
in competitive ratio for BSTs over the trivial $O(\log n)$ upper
bound.

Being $O(\log \log n)$-competitive, tango trees are always at most a factor
$O(\log \log n)$ worse than ${\rm OPT}(X)$.
On the other hand, they may actually pay this multiplicative overhead at
each access, implying that they have $\Theta(\log \log n \log n)$
worst-case access time, and use $\Theta(m \log \log n \log n)$ time on some
access sequences of length $m$. In comparison, any balanced BST (even
static) has $O(\log n)$ worst-case access time and spends $O(m \log n)$ on
every access sequence.

The problem we consider in this paper is whether it is possible to combine the
best of these bounds---that is, whether an $O(\log \log n)$-competitive BST
algorithms that performs each access in optimal $O(\log n)$ worst-case time
exists. We answer it affirmatively by presenting a data structure achieving
these complexities. It is based on the overall framework of tango
trees---however, where tango trees use red-black trees~\cite{redblack} for
storing what is
called preferred paths, we develop a specialized BST representation of the
preferred paths, tuned to the purpose. This representation is the main
technical contribution, and its description takes up the bulk of the paper.  

In the journal version of their seminal paper on tango trees,
Demaine~\emph{et al.}\ suggested that such a structure
exists. Specifically, in the further work section, the authors gave a short
sketch of a possible solution. Their suggested approach, however, relies on
the existence of a BST supporting dynamic finger, split and merge in O(log
r) worst-case time where $r$ is $1$ plus the rank difference between the
accessed element and the previously accessed element. Such a BST could
indeed be used for the auxiliary tree representation of preferred paths.
However, the existence of such a structure (in the BST-model) is an open
problem.  Consequently, since the publication of their work, the authors
have revised their stance and consider the problem solved in this paper to
be an open problem \cite{john}.  Recently, Woo~\cite{woo} made some
progress concerning the existence of a BST having the dynamic finger
property in worst-case. He developed a BST algorithm satisfying, based on
empirical evidence, the dynamic finger property in worst-case.
Unfortunately this BST algorithm does not allow insertion/deletion or
split/merge operations, thus it cannot be used to maintain the preferred
paths in a tango tree.

After the publication of the tango tree paper, two other $O(\log \log
n)$-competitive BSTs have been introduced by Derryberry~\emph{et
al.}~\cite{multisplay,MultisplayThesis} and
Georgakopoulos~\cite{loglognsplay}.  The multi-splay
trees~\cite{multisplay} are based on tango trees, but instead of using
red-black trees as auxiliary trees, they use splay trees~\cite{splay}. As a
consequence, multi-splay trees can be
shown~\cite{multisplay,MultisplayThesis} to satisfy additional properties,
including the scanning and working-set bounds of splay trees, while
maintaining $O(\log \log n)$-competitiveness.  Georgakopoulos uses the
interleave lower bound of Demaine {\em et al.} to develop a variation of
splay trees called {\em chain-splay} trees that achieves $O(\log \log
n)$-competitiveness while not maintaining any balance condition explicitly.
However, neither of these two structures achieves a worst-case single
access time of $O(\log n)$. A data structure achieving the same running
time as tango trees alongside $O(\log n)$ worst-case single access time was
developed by Kujala and Elomaa~\cite{poketree}, but this data structure
does not adhere to the BST model (in which the lower bounds on ${\rm
OPT}(X)$ are proved).

The rest of this paper is organized as follows: In
Section~\ref{preliminaries}, we formally define the model of BSTs
and the access sequences considered. We state the lower bound on ${\rm
OPT}(X)$ developed in~\cite{tango,wilber} for analyzing the competitive
ratio of BSTs. We also describe the central ideas of tango trees.  In
Section~\ref{hybrid}, we introduce a preliminary data structure called
\emph{hybrid trees}, which does not fit the BST model proper, but which is
helpful in giving the main ideas of our new BST structure. Finally in
Section~\ref{zippertree}, we develop this structure further to fit the BST
model. This final structure, called \emph{zipper trees}, is a BST achieving
the optimal worst-case access time while maintaining the $O(\log \log
n)$-competitiveness property.


\section{Preliminaries}
\label{preliminaries}

\subsection{BST Model}
\label{model}
In this paper we use the binary search tree model (BST model) defined
by Wilber~\cite{wilber}, which is standard in the area.  Each node stores a
key from a totally ordered universe, and the keys obey in-order: at any
node, all of the keys in its left subtree are less than the key stored in
the node, and all of the keys in its right subtree are greater (we assume
no duplicate keys appear). Each node has three pointers, pointing to its
left child, right child, and parent. Each node may keep a
constant\footnote{According to standard conventions, $O(\log_2 n)$ bits are
considered as constant.}  amount of additional information, but no further
pointers may be used.

To perform an access, we are given a pointer initialized to the root. An access
consists of moving this pointer from a node to one of its adjacent nodes
(through the parent pointer or one of the children pointers) until it reaches
the desired element. Along the way, we are allowed to update the fields and
pointers in any nodes that the pointer touches. The access cost is the number
of nodes touched by the pointer.  

As is standard in the area, we only consider sequences consisting of element
accesses on a fixed set $S$ of $n$~elements. In particular, neither
unsuccessful searches, nor updates appear.

\subsection{Interleave Lower Bound}
The interleave bound is a lower bound on the time taken by any binary
search tree in the BST model to perform an access sequence
$X=\{x_1,x_2,\ldots,x_m\}$. The interleave bound was developed by
Demaine~\emph{et al.}~\cite{tango} and was derived from a previous bound of
Wilber~\cite{wilber}.

Let $P$ be a static binary search tree of minimum height, built on the set
of keys $S$. We call $P$ the \emph{reference} tree. For each node $y$ in
$P$, we consider the accesses $X$ to keys in the nodes in the subtree of
$P$ rooted at $y$ (including $y$). Each access of this subsequence is then
labelled ``left'' or ``right'', depending on whether the accessed node is
in the left subtree of $y$ (including $y$), or in its right subtree,
respectively. The \emph{amount of interleaving through $y$} is the number
of alternations between left and right labels in this subsequence. The
interleave bound ${\rm IB}(X)$ is the sum of these interleaving amounts
over all nodes $y$ in $P$.
The exact statement of the lower bound from~\cite{tango} is as follows:
\begin{thm}
\label{ib}
For any access sequence $X$, 
${\rm IB}(X)/2- n$ is a lower bound on ${\rm OPT}(X)$.
\end{thm}

\subsection{Tango Trees}
We outline the main ideas of tango trees~\cite{tango}.
As in the previous section, denote by the reference tree $P$ a static
binary search tree of height $O(\log n)$ built on a set of keys $S$. The
\emph{preferred child} of an internal node $y$ in $P$ is defined as its
left or right child depending on whether the last access to a node in the
subtree rooted at $y$ (including $y$) was in the left subtree of $y$
(including $y$) or in its right subtree respectively. We call a maximal
chain of preferred children a \emph{preferred path}. The set of preferred
paths naturally partitions the elements of $S$ into disjoint subsets of
size $O(\log n)$ (see the left part of Figure~\ref{fig-tango}). Remember
that $P$ is a static tree, only the preferred paths may evolve over time
(after each access).

The ingenious idea of tango trees is to represent the nodes on a preferred
path as a balanced \emph{auxiliary} tree of height $O(\log \log n)$. The
tango tree can be seen as a collection of auxiliary trees linked
together. The leaves of an auxiliary tree representing a preferred path $p$
link to the root of auxiliary trees representing the paths immediately
below $p$ in $P$ (see Fig.~\ref{fig-tango}), with the links uniquely
determined by the inorder ordering. The auxiliary tree containing the root
of $P$ constitutes the top-part of the tango tree. In order to distinguish
auxiliary trees within the tango tree, the root of each auxiliary tree is
marked (using one bit).

\begin{figure}
    \begin{center}
        \includegraphics[width=0.6\textwidth]{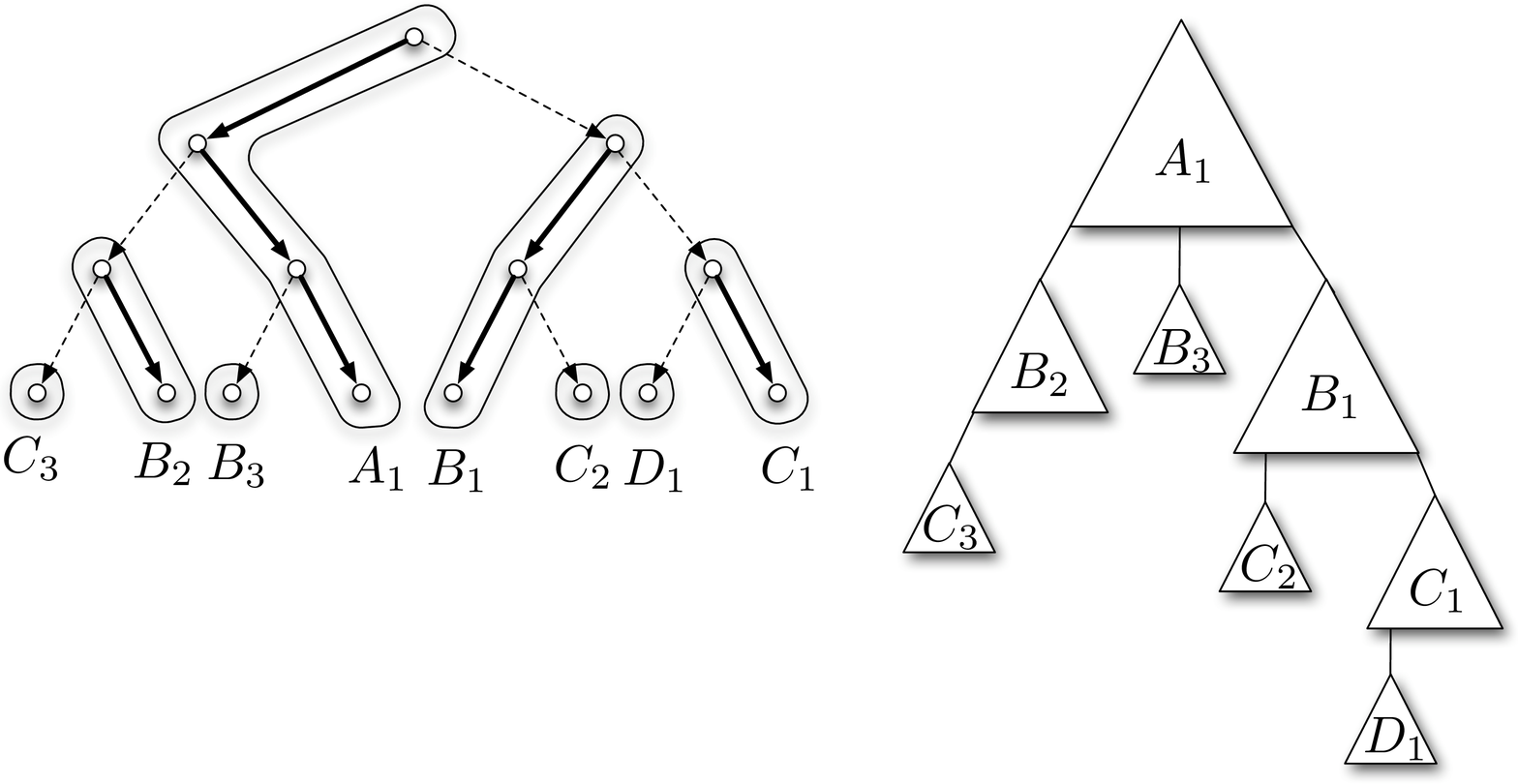}
    \end{center}
    \caption{\label{fig-tango} On the left, reference tree P with its
    preferred paths. On the right, the tango tree representation of P.}
\end{figure}

\todo{Why are the same letters used repeatedly for different subtrees in
Fig.~\ref{fig-tango}?  Rolf.}

Note that the reference tree $P$ is not an explicit part of the structure,
it just helps to explain and understand the concept of tango trees. When an
access is performed, the preferred paths of $P$ may change. This change is
actually a combination of several cut and concatenation operations
involving subpaths. Auxiliary trees in tango tree are implemented as
red-black trees~\cite{redblack}, and~\cite{tango} show how to implement
these cut and concatenation operations using standard split and join
operations on red-black tree. Here are the main two operations used to
maintain tango trees:
\begin{itemize}
\item C{\scriptsize UT}-T{\scriptsize ANGO}($A$, $d$) -- cut the red-black tree $A$ into two red-black trees, one storing the path of all nodes of depth at most $d$, and the other storing the path of all nodes of depth greater than $d$.
\item C{\scriptsize ONCATENATE}-T{\scriptsize ANGO}($A$, $B$) -- join two red-black trees that store two disjoint paths where the bottom of one path (stored in $A$) is the parent of the top of the other path (stored in $B$). So the root of $B$ is attached to a leaf of $A$. 
\end{itemize}

These operations take $O(\log k)$ time for trees of size $k$ using extra
information stored in nodes. As the trees store paths in $P$, we have $k =
O(\log n)$. In addition to storing the key value and the depth in $P$, each
node stores the minimum and maximum depth over the nodes in its subtree
within its auxiliary tree. This additional data can be trivially maintained
in red-black trees with a constant-factor overhead.

Hence, if an access passes $i$ different preferred paths in $P$, the
necessary change in the tango tree will be $O(i)$ cut and concatenation
operations, which is performed in $O(i\log \log n)$ time. Over an access
sequence $X$ the total number of cut and concatenation operations performed
in $P$ corresponds to the interleave bound $O({\rm IB}(X))$, thus tango
tree performs this access sequence in $O(\log \log n \, {\rm IB}(X))$ time.

\section{Hybrid Trees}
\label{hybrid}

In this section, we introduce a data structure called \emph{hybrid trees},
which has the right running time, but which does not fit the BST model
proper. However, it is helpful intermediate step which contains the main
ideas of our final BST structure.

\subsection{Path Representation}

For all preferred paths in $P$, we keep the top $\Theta(\log \log n)$ nodes
exactly as they appear on the path. We call this the \emph{\toppath{}}. The
remaining nodes (if any) of the path we store as a red-black tree, called
the \emph{\bottomtree}, which we attach below the \toppath{}. Since a
preferred path has size $O(\log n)$, this \bottomtree{} has height $O(\log
\log n)$.
More precisely, we will maintain the invariant that a \toppath{} has length
in $[\log\log n, 3\log\log n]$, unless no \bottomtree{} appears, in
which case the constraint is  $[0,3\log\log n]$. (This latter
case, where no \bottomtree{} appears, will induce simple and obvious
variants of the algorithms in the remainder of the paper, variants which we
for clarity of exposition will not mention further.)

A \emph{hybrid tree} consists of all the preferred paths of $P$,
represented as above, linked together to form one large tree, analogous
to tango trees.

The required worst-case search complexity of hybrid trees is captured by
the following lemma.
\begin{lemma}
\label{height}
A hybrid tree $T$ satisfies the following property:
$$
d_T(x)= O(d_P(x)) \quad \forall x \in S,
$$ 
where $d_T(x)$ and $d_P(x)$ is defined as the depth of the node $x$ in the
tree $T$ and in the reference tree $P$, respectively. In particular, $T$
has $O(\log n)$ height.
\end{lemma}
\begin{proof}
Consider a preferred path $p$ in $P$ and its representation tree $h$. The
distance, in terms of number of links to follow, from the root of $h$ to one
of its nodes or leaves $x$ is no more than a constant times the distance
between $x$ and the root of $p$. Indeed, if $x$ is part of the \toppath{},
then the distance to the root of the path by construction is the same in
$h$ and $p$. Otherwise, this distance increases by at most a constant
factor, since $h$ has a height of $O(\log \log n)$ and the distance in $p$
is already $\Omega(\log \log n)$.

Since the number of links followed between preferred paths is the same in
$P$ and $T$, the lemma follows. \todo{I (Rolf) here removed an argument I
didn't see the need for. Karim, please check.\\
(Karim) That's correct.}
\end{proof}

\subsection{Maintaining Hybrid Trees under Accesses}
\label{maintaining-hybrid-trees}

Like in tango trees, the path $p$ traversed in $P$ to reach a desired node
may pass through several preferred paths. During this access the preferred
paths in $P$ must change such that $p$ becomes the new preferred path
containing the root. This is performed by cut and concatenate operations on
the preferred paths passed by $p$. When $p$ leaves a preferred path, this
is cut at a depth corresponding to the depth in $P$ of the point of leave
of the preferred path, and the top part cut out is concatenated with the
next preferred path to be traversed.

We note that the algorithm may as well restrict itself to cutting when
traversing $p$, producing a sequence of cut out parts hanging below each
other, which can then be concatenated in one go at the end, producing the
new preferred path starting at the root. We will use this version below.

In this subsection, we will show how to maintain the hybrid tree
representation of the preferred paths after an access. Our goal is to
describe how to perform the operations \emph{cut} and \emph{concatenate} in
the following complexities: When the search path passes only the \toppath{}
of a preferred path, the cut procedure takes $O(k)$ time, where $k$ is the
number of nodes traversed in the \toppath{}. When the search path passes
the entire \toppath{} and ends up in the \bottomtree{}, the cut procedure
takes $O(\log \log n )$ time.  The concatenation operation, which glues
together all the cut out path representation parts at the end of the
access, is bounded by the time used by the search and the cut operations
performed during the access.

Assuming these running times, it follows, by the invariant that all
\toppath{}s (with \bottomtree{}s below them) have length $\Theta(\log\log
n)$, that the time of an access involving $i$ cut operations in $P$ is
bounded both by the number of nodes on the search path~$p$, and by $i \log
\log n$. By Lemma~\ref{height}, this is $O(\min\{\log n, i \log \log n\})$
time. Hence, we will have achieved optimal worst-case access time while
maintaining $O(\log \log n)$-competitiveness.

\paragraph{CUT:} Case~1: We only traverse the \toppath{} of a path
representation. Let $k$ be the number of nodes traversed in this \toppath{}
and let $x$ be the last traversed node in this \toppath{}. The cut
operation marks the node succeeding $x$ on the \toppath{} as the new root of the
path representation, and unmarks the other child of $x$.

The cut operation now has removed $k$ nodes from the \toppath{} of the path
representation. This implies that we possibly have to update the
representation, since the $\Theta(\log \log n)$ bound on the size of its
\toppath{} has to be maintained. Specifically, if the size of the top path
drops below $2 \log\log n$, we will move some nodes from the \bottomtree{}
to the \toppath{}. The nodes should be those from the \bottomtree{} having
smallest depth (in $P$), i.e., the next nodes on the preferred path in
$P$. After a cut of $k$ nodes it is for small $k$ (smaller than $\log\log
n$) not clear how to extract the next $k$ nodes from the \bottomtree{} in
$O(k)$ time. Instead, we use an \emph{extraction process}, described below,
which extracts the next $\log \log n$ nodes from the \bottomtree{} in
$O(\log \log n)$ steps and run this process incrementally: Whenever further
nodes are cut from the \toppath{}, the extraction process is advanced by
$\Theta(k)$ steps, where $k$ is the number of nodes cut, and then the
process is stopped until the next cut at this path occurs. Thus, the work
of the extraction process is spread over several Case~1 cuts (if not
stopped before by a Case 2 cut, see below). The speed of the process is
chosen such that the extraction of $\log \log n$ nodes is completed before
that number of nodes have been cut away from the \toppath, hence it will
raise the size of the \toppath{} to at least $2 \log\log n$ again. In
general, we maintain the additional invariant that the \toppath{} has size
at least $2 \log\log n$, unless an extraction process is ongoing. For
larger values of $k$ (around $\log\log n$), up to two extraction processes
(the first of which could be partly done by a previous access) will be used
to ensure this.

Case 2: We traverse the entire \toppath{} of path representation $A$, and
enter the \bottomtree{}. Let $x$ be the last traversed node in $A$ and let
$y$ be the marked child of $x$ that is the root of the next path
representation on the search path. First, we finish any pending extraction
process in $A$, so that its \bottomtree{} becomes a valid red-black
tree.\todo{This could move $x$ onto the top path, which does not seem
anticipated by the existing phrasing. Hence the changed phrasing. Has
implications for concatenate, see Todo below. Rolf.} Then we rebuild the
\toppath{} into a red-black tree in linear time (details appear under the
description of concatenate below), and we join it with the \bottomtree{}
using C{\scriptsize ONCATENATE}-T{\scriptsize ANGO}.  Then we perform
C{\scriptsize UT}-T{\scriptsize ANGO}($A'$, $d$) where $A'$ is the combined
red-black tree, and $d=d_P(y)-1$. After this operation, all nodes of depth
greater than $d$ are removed from the path representation $A$ to form a new
red-black tree $B$ attached to $A$ (the root of $B$ is marked in the
process). To make the tree~$B$ a valid path representation, we perform an
extraction process twice, which extracts $2\log \log n$ nodes from it to
form a \toppath. Finally we unmark $y$. This takes $O(\log \log n)$ time in
total.

\paragraph{CONCATENATE:}

What is cut out during an access is a sequence of \toppath{}s (case~1 cuts)
and red-black trees (case 2 cuts) hanging below each other. We have to
concatenate this sequence into one path representation. We first rebuild
all sequences of consecutive subpaths (maximum sequences of nodes which
have one marked child) into valid red-black trees, in time linear in the
number of nodes of each sequence (details below). This leaves a sequence of
valid red-black trees hanging below each other. Then we iteratively perform
C{\scriptsize ONCATENATE}-T{\scriptsize ANGO}($A$,$B$), where $A$ is the
current highest red-black tree and $B$ is the tree hanging below $A$, until
there is one remaining red-black tree. Finally we extract $2\log \log n$
nodes from the obtained red-black tree to construct the \toppath{} of the path
representation. The time used for concatenate is bounded by the time used
already during the search and cut part of the access.

\todo{(Rolf)Karim: I here removed the note on efficiency by reusing the toppath,
as this may no longer exists, by the change of Case 2, as mentioned in the
last Todo above. Please check. Rolf.\\
(Karim) I think we still have that the sequence of cut out parts is composed first by
$\Theta(\log \log n)$ nodes forming a path. Anyway we will discuss that for the final 
version.}

One way to convert a path of length $k$ into a red-black tree in $O(k)$
time is as follows: consider each node on the path as a red-black tree of
size one. We iteratively perform a series of C{\scriptsize
ONCATENATE}-T{\scriptsize ANGO}($A$,$B$) operations for each pair of
red-black trees $A$ followed by $B$. After each iteration the number of
trees is divided by 2 and their size is doubled, giving a total time for
rebuilding a path into a valid red-black tree of $O(\sum_{i=1}^{\log k}
ik/2^i)=O(k)$.

\paragraph{EXTRACT:}


We now show how to perform the central process of our structure, namely
extracting the next part of a \toppath{} from a
\bottomtree{}. Specifically, we will extract a subpath of $\log \log n$
nodes of minimum depth (in $P$) from the \bottomtree{}~$A'$ of a given path
representation $A$, using $O(\log \log n)$ time.

Let $x$ be the deepest nodes on the \toppath{} of $A$, such that the
unmarked child of $x$ corresponds to the root of the \bottomtree{}
$A'$. The extraction process will separate the nodes of depth (in $P$)
smaller than $d=d_P(x)+\log \log n$ from the \bottomtree{}~$A'$.
Let a \emph{zig} segment of a preferred path $p$ be a maximal sequence of
nodes such that each node in the sequence is linked to its right child in
$p$. A \emph{zag} segment is defined similarly such that each node on the
segment is linked to its left child (see Fig.~\ref{fig-zig-zag}).

The key observation we exploit is the following: the sequence of all zig
segments, ordered by their depth in the path, followed by the sequence of
all reversed zag segments, ordered reversely by their depth in the path, is
equal to the ordering of the nodes in key space (see
Fig.~\ref{fig-zig-zag}).
This implies that to extract the nodes of depth smaller than $d$ (in $P$)
from a \bottomtree{}, we can cut the extreme ends (in key space) of the
tree, linearize them to two lists, and then combine them by a binary merge
procedure using depth in $P$ as the ordering. This forms the core of the
extract operation.

\begin{figure}
    \begin{center}
        \includegraphics[width=0.4\textwidth]{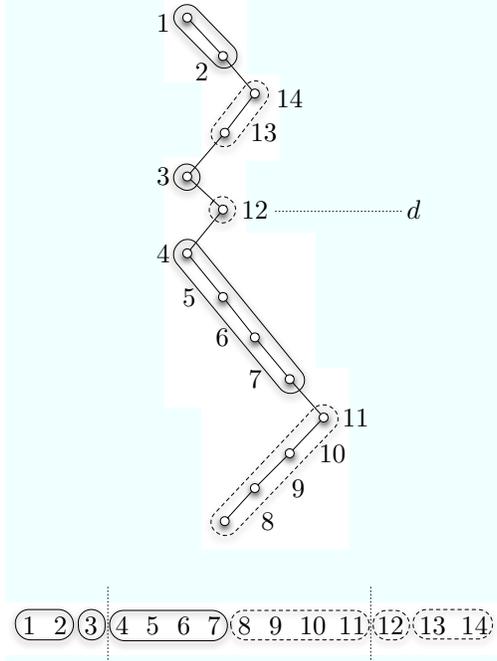}
    \end{center}
    \caption{\label{fig-zig-zag} A path, its decomposition into zig (solid
    regions) and zag (dashed regions) segments, and its layout in key order.}
\end{figure}

We have to do this using rotations, while maintaining a tree at all
times. We now give the details of how to do this, with
Fig.~\ref{fig-extract} illustrating the process.

Using extra fields of each node storing the minimum and maximum depth value
(in $P$) of nodes inside its subtree, we can find the node $\ell'$ of
minimum key value that has a depth greater than $d$ in $O(\log \log n)$
time, by starting at the root of $A'$ and repeatedly walking to the
leftmost child whose subtree has a node of depth greater than $d$. Then
define $\ell$ as the predecessor of $\ell'$. Symmetrically, we can find the
node $r'$ of maximum key value that has depth greater than $d$ and define
$r$ as the successor of $r'$.

First we \todo{Operation split should be introduced properly, with BST
formulation. Probably in tango section. Rolf.} split $A'$ at $\ell$ to
obtain two subtrees $B$ and $C$ linked to the new root $\ell$ where $B$
contains a first sequence of nodes at depth smaller than $d$. Then we split
$C$ at $r$ to obtain the subtrees $D$ and $E$ where $E$ contains a second
sequence of nodes at depth smaller than $d$.

In $O(\log \log n)$ time we convert the subtrees $B$ and $E$ into paths
corresponding to an ordered sequences of zig segments for $B$ and zag
segments for $E$. To do so we perform a left rotation at the root of $B$
until its right subtree is a leaf (i.e., when its right child is a marked
node). Then we repeat the following: if the left child of the root has no
right child the we perform a right rotation at the root of $B$ (adding one
more node to right spine, which will constitute the final path). Otherwise
we perform a left rotation at the left child of the root of $B$, moving its
right subtree into the left spine. This process takes a time linear in the
size of $B$, since each node is involved in a rotation at most 3 times
(once a node enters the left spine, it can only leave it by being added to
the right spine). A symmetric process is performed to convert the subtree
$E$ into a path.

The last operation, called a \emph{zip}, merges (in term of depths in $P$)
the two paths $B$ and $E$, in order to form the next part of the
\toppath. We repeatedly select the root of $B$ or $E$ that has the smallest
depth in the tree $P$. The selected root is brought to the bottom of the
\toppath{} using $O(1)$ rotations. The zip operation stops when the subtrees
$B$ and $E$ are both empty. Eventually, we perform a left rotation at the
node $\ell$ if needed, i.e., if $r$ has a smaller depth in $P$ than $\ell$.

\begin{figure}
    \begin{center}
        \includegraphics[width=0.9\textwidth]{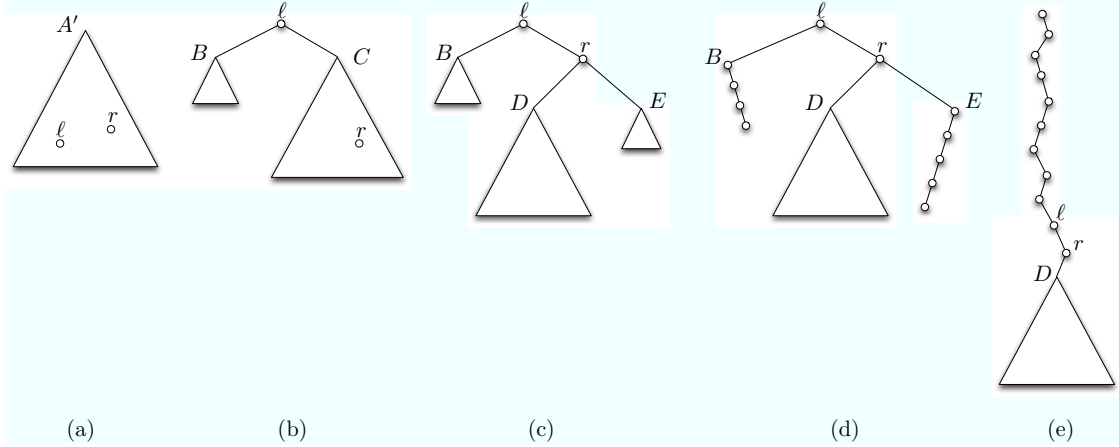}
    \end{center}
    \caption{\label{fig-extract} (a) Tree $A'$. (b) Split $A'$ at $\ell$. (c) Split $C$ at $r$. (d) Convert the subtrees $B$ and $E$ into paths. (e) Zip the paths $B$ and $E$.}
\end{figure}

The time taken is linear in the extracted number of nodes, i.e, $\log \log
n$. The process consists of a series of rotations, hence can stopped and
resumed without problems. 

Therefore, the discussion presented in this section allows us to conclude with
the following theorem.

\begin{thm}
Our hybrid tree data structure is $O(\log \log n)$-competitive and performs
each access in $O(\log n)$ worst-case time.
\end{thm}


\subsection{Hybrid Trees and the BST Model} 

We specify in the description of the cut operation (more precisely, in
case~1) that the extraction process is executed incrementally, i.e., the
work is spread over several cut operations. In order to efficiently revive 
an extraction process which has been stopped at some point in the past, we
have to return to the position where its next rotation should take
place. This location is unique for each path representation, and is always
in its \bottomtree{}. Thus, traversing its \toppath{} to reach the
\bottomtree{} would be too costly for the analysis of case~1. Instead, we
store in the marked node (the first node of the \toppath{}) appropriate
information on the state of the process. Additionally, we store an extra
pointer pointing to the node where the next rotation in the process should
take place. This allows us to revive an extraction process in constant
time. Unfortunately, the structure so obtained is not in the BST model (see
Section~\ref{model}), due to the extra pointer. In the next section we show
how to further develop the idea from this section into a data structure
fitting the BST model.

Still, we note that the structure of this section can be implemented in the
comparison based model on a pointer machine, with access sequences $X$
being served in $O(\log \log n \, {\rm OPT}(X))$ time, and each access
taking $O(\log n)$ time worst-case.

\section{Zipper Trees}
\label{zippertree}

The data structure described in the previous section is a BST, except that
each marked node has an extra pointer facilitating constant time access to
the point in the path representation where an extraction process should be
revived. In this section, we show how to get rid of this extra pointer and
obtain a data structure with the same complexity bounds, but now fitting
the BST model described in Section~\ref{model}. To do so, we develop a more
involved version of the representation of preferred paths and the
operations on them.
The goal of this new path representation is to ensure that all rotations of
an extraction process are located within distance $O(1)$ of the root of the
tree of the representation. The two main ideas involved are: 1) storing the
\toppath{} as lists,
hanging to the sides of the root, from which 
the \toppath{} can be generated incrementally by merging as it is traversed
during access, and 2) using a version of the split operations that only
does rotations near the root. The time complexity analysis follows that of
hybrid trees, and will not be repeated.

\subsection{Path Representation} For all preferred paths in $P$ we
decompose its highest part
into two sequences, containing its zig and its zag segments,
respectively. These are stored as two paths of nodes, of increasing and
decreasing key values, respectively. As seen in
Section~\ref{maintaining-hybrid-trees} (cf.\ Fig.~\ref{fig-zig-zag}), both
will be ordered by their depth in $P$. Let $\ell$ and $r$ be the highest
node in the zig and zag sequence respectively. The node $\ell$ will be the
root of the auxiliary tree (the marked node). The remainder of the zig
sequence is the left subtree of $\ell$, $r$ is its right child, and the
remainder of the zag sequence is the right subtree of $r$. We call this
upper part of the tree a \emph{zipper}. We repeat this decomposition once
again for the next part of the path
to obtain another zipper which is the left subtree of $r$. Finally the
remaining of the nodes on the path are stored as a red-black tree of height
$O(\log \log n)$, hanging below the lowest zipper. Fig.~\ref{fig-zipper}
illustrates the construction. The two zippers constitute the
\emph{\toppath}, and the red-black tree the \emph{\bottomtree}. Note that
the root of the \bottomtree{} is reachable in $O(1)$ time from the root of
the path representation. We will maintain the invariant that individually,
the two zippers contain at most $\log\log n$ nodes each, while (if the
\bottomtree is non-empty) they combined contain at least $(\log\log
n)/2$ nodes.

A \emph{zipper tree} consists of all the preferred paths of $P$,
represented as above, linked together to form one large tree.

\begin{figure}
    \begin{center}
        \includegraphics[width=0.45\textwidth]{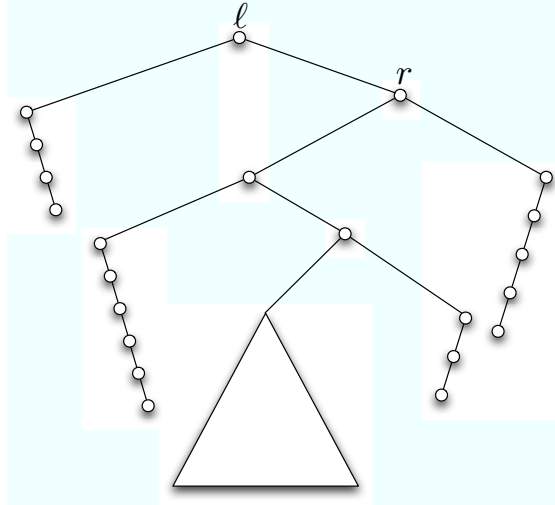}
    \end{center}
    \caption{\label{fig-zipper} The path representation in zipper trees.}
\end{figure}

\subsection{Maintaining Zipper Trees under Accesses}

We now give the differences, relative to
Section~\ref{maintaining-hybrid-trees}, of the operations during an access.

\paragraph{CUT:}
When searching a path representation, we incrementally perform a zip
operation (i.e., a merge based on depth order) on the top zipper, until
it outputs either the node searched for, or a node that leads to the next
path representation. If the top zipper gets exhausted, the lower zipper
becomes the upper zipper, and an incremental creation of a new lower zipper
by an extraction operation on the \bottomtree{} is initiated (during which
the lower zipper is defined to have size zero). Each time one more node
from the top zipper is being output (during the current access, or during a
later access passing through this path representation), the extraction
advances $\Theta(1)$ steps. The speed of the extraction process is chosen
such that it finishes with $\log\log n$ nodes extracted before $(\log\log
n)/2$ nodes have been output from the top zipper. The new nodes will make
up a fresh lower zipper, thereby maintaining the invariant.

If the access through a path representation overlaps (in time) at most one
extraction process (either initiated by itself or by a previous access), it
is defined as a case~1 cut. No further actions takes place, besides the
proper remarkings of roots of path representations, as in
Section~\ref{maintaining-hybrid-trees}. If a second extraction process is
about to be initiated during an access, we know that $\Theta(\log\log n)$
nodes have been passed in this path representation, and we define it as a
case~2 cut. Like in Section~\ref{maintaining-hybrid-trees} this now ends by
converting the path representation to a red-black tree, cutting it like in
tango trees, and then converting the red-black tree remaining into a valid
path representation (as defined in the current section), all in
$\Theta(\log\log n)$ time.

\paragraph{CONCATENATE:} There is no change from
Section~\ref{maintaining-hybrid-trees}, except that the final path
representation produced is as defined in the current section.

\paragraph{EXTRACT:} The change from Section~\ref{maintaining-hybrid-trees}
is that the final zip operation is not performed (the process stops at
step~(d) in Fig.~\ref{fig-extract}), and that we must use a search and a
split operation on red-black trees where all structural changes consist of
rotations a distance $O(1)$ from the root\footnote{As no actual details of
the split operation used are given in \cite{tango}, we do not know whether
their split operation fulfills this requirement. It is crucial for our
construction that such a split operation is possible, so we describe one
solution here.} (of the \bottomtree{}, which is itself at a distance $O(1)$
from the root of the zipper tree). Such a split operation is described in
the appendix (Part I). Note that searching takes place incrementally as part of the
split procedure.

\section{Conclusion}

The main goal in this area of research is to improve the competitive ratio of $O(\log \log n)$. Here we have been able to tighten other bounds, namely the worst-case search time. We think this result helps providing a better understanding of competitive BSTs. It could be that competitiveness is in conflict with balance maintenance, i.e., an $O(1)$-competitive binary search tree could possibly not  guarantee an $O(\log n)$ worst-case search time. For instance splay-tree~\cite{splay} and GreedyFuture tree \cite{munro2000competitiveness, geoBST}, the two BSTs that are conjectured to be dynamically optimal, do not guarantee optimal worst-case search time. Thus even if dynamically optimal trees exist, our result could still be a good alternative with optimal worst-case performance. 

We also think that the ideas developed to achieve our result have their own interest. They can be used to improve the worst-case performance of a data structure while maintaining the same amortized performance. For example we show in the appendix (Part II) how to adapt them in order to improve the worst-case running time of the multipop operation on a stack. 

\bibliographystyle{abbrv}    
\bibliography{biblio}

\begin{thebibliography}{10}

\bibitem{Cormen2001}
T.~H. Cormen, C.~E. Leiserson, R.~L. Rivest, and C.~Stein.
\newblock {\em Introduction to Algorithms}.
\newblock MIT Press, 2001. Section 10.1: Stacks and queues, pages 200--204.

\bibitem{tango}
E.~Demaine, D.~Harmon, J.~Iacono, and M.~P\u{a}tra\c{s}cu.
\newblock Dynamic optimality---almost.
\newblock {\em SICOMP}, 37(1):240--251, 2007.
\newblock Also in \textit{Proc. FOCS 2004}.

\bibitem{geoBST}
E.~D. Demaine, D.~Harmon, J.~Iacono, D.~M. Kane, and M.~Patrascu.
\newblock The geometry of binary search trees.
\newblock In {\em Proc. of the 10th ACM-SIAM Symp. on Disc. Alg. (SODA)}, pages
  496--505, 2009.

\bibitem{multisplay}
J.~Derryberry, D.~D. Sleator, and C.~C. Wang.
\newblock ${O}(\log \log n)$-competitive dynamic binary search trees.
\newblock In {\em Proc. of the 7th ACM-SIAM Symp. on Disc. Alg. (SODA)}, pages
  374--383, 2006.

\bibitem{loglognsplay}
G.~F. Georgakopoulos.
\newblock Chain-splay trees, or, how to achieve and prove
  loglogn-competitiveness by splaying.
\newblock {\em Inf. Process. Lett.}, 106(1):37--43, 2008.

\bibitem{redblack}
L.~J. Guibas and R.~Sedgewick.
\newblock A dichromatic framework for balanced trees.
\newblock In {\em Proc. of the 19th Found. of Comp. Sci. (FOCS)}, pages 8--21,
  1978.

\bibitem{john}
J.~Iacono.
\newblock Personal communication, July 2009.

\bibitem{poketree}
J.~Kujala and T.~Elomaa.
\newblock Poketree: {A} dynamically competitive data structure with good
  worst-case performance.
\newblock In {\em Proc. {ISAAC} 2006}, volume 4288 of LNCS, pages 277--288,
  2006.

\bibitem{munro2000competitiveness}
J.~I. Munro.
\newblock On the competitiveness of linear search.
\newblock In {\em Proc. of the 8th Annual European Symposium (ESA)}, volume
  1879 of LNCS, pages 338--345, 2000.

\bibitem{splay}
D.~D. Sleator and R.~E. Tarjan.
\newblock Self-adjusting binary search trees.
\newblock {\em J. ACM}, 32(3):652--686, 1985.

\bibitem{MultisplayThesis}
C.~C. Wang.
\newblock {\em Multi-Splay Trees}.
\newblock PhD thesis, Computer Science Department, School of Computer Science,
  Carnegie Mellon University, 2006.

\bibitem{wilber}
R.~Wilber.
\newblock Lower bounds for accessing binary search trees with rotations.
\newblock {\em SICOMP}, 18(1):56--67, 1989.

\bibitem{woo}
S.~L. Woo.
\newblock {\em Heterogeneous Decomposition of Degree-Balanced Search Trees and
  Its Applications}.
\newblock PhD thesis, Computer Science Department, School of Computer Science,
  Carnegie Mellon University, 2009.

\end{thebibliography}

\newpage
\section*{APPENDIX}

\section*{Part I}

We present the split operation on red-black trees where all structural 
changes consist of rotations a distance $O(1)$ from the root.The end 
result should be a tree where the root is the node $x$ split on, 
with a left child that is a red-black tree on the nodes smaller than $x$,
and a right child that is a red-black tree on the nodes larger than
$x$. Red-black trees are binarized (2,4)-trees, i.e., multi-way nodes are
substituted by small, perfectly balanced binary search trees, with the
color annotation keeping track of the boundaries between them. For ease of
exposition, we describe the split process partly in (2,4)-tree terms.

The overall idea of a normal split operation on (2,4)-trees is to follow
the search path towards~$x$ and cleave (2,4)-tree nodes passed, then later
glue the two parts of cleaved nodes (now maybe of too low degree) to their
level-wise siblings (now maybe of too high degree), and then split (in the
sense of rebalancing of (2,4)-nodes) these if necessary. By cleaving a
(2,4)-node, we mean dividing the keys (here, binary nodes) inside the
(2,4)-node, and the subtrees of the (2,4)-node, into two parts, based on
whether they order-wise are smaller or larger than $x$.

Normally, cleaving is top-down, and gluing and splitting is bottom-up.
However, to keep the working point, where the rotations take place, fixed
at the root---and maintain a BST at all times---we instead fold the two
sides of the cleaved path around during searching-and-cleaving (treating
each side of the cleaved path sort of as a rope in a pulley, the pulleys
being the working point), while keeping track of heights (in (2,4)-tree
terms) of subtrees. Later, the gluing process is done in a reverse action
(running the ropes the other way), while making use of the heights of
subtrees.

We now describe the details of the cleaving process. We assume wlog.\ that
the search path initially proceeds to the left child of the root (the other
case being symmetric). During the cleaving process, we maintain the
following shape of the binary tree: Let $P$ be the search path of the
initial red-black tree traversed so far, and let $P_{(2,4)}$ be the
(2,4)-nodes touched by $P$. The top of the right spine (including the root)
consists of the binary nodes in $P_{(2,4)}$ whose keys are larger than
$x$. The subtrees hanging from these nodes are the subtrees of $P_{(2,4)}$
which are larger than $x$.  The top of the left spine (excluding the root)
consists of the binary nodes in $P_{(2,4)}$ whose keys are smaller than
$x$. The subtrees hanging from these nodes are the subtrees of $P_{(2,4)}$
which are smaller than $x$, except for the topmost of the subtrees, which
is the remaining part of the initial tree. See
Fig.~\ref{fig-split-from-root}.

\begin{figure}
    \begin{center}
        \includegraphics[width=0.9\textwidth]{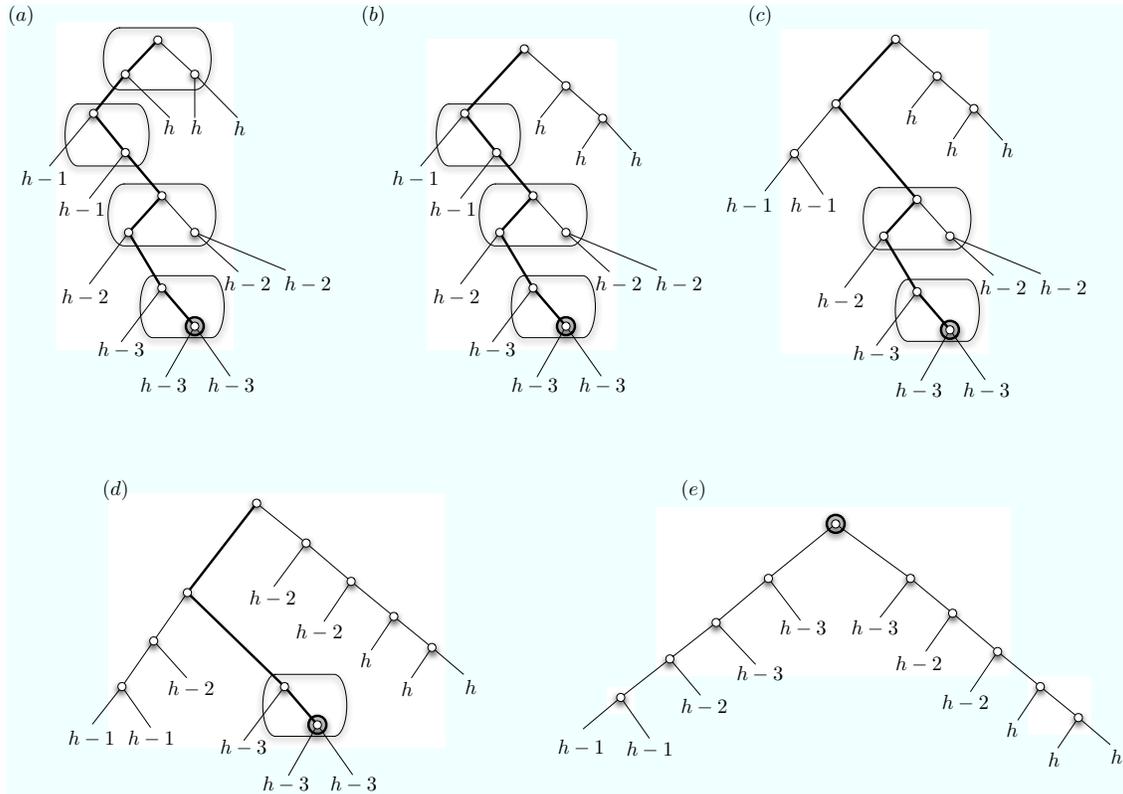}
    \end{center}
    \caption{\label{fig-split-from-root} A split operation (first half of
    the operation) using rotations close to the root. Rounded boxes
    delineate $(2,4)$-nodes, values $h$, $h-1$, etc., designate subtrees of
    that $(2,4)$-tree height. The thick line is the search path, ending at
    the circled node.}
\end{figure}

Advancing the cleaving process by one more (2,4)-node of the search path
means rebuilding a subtree consisting of at most the root, its left child,
and the up to three binary nodes of the next (2,4)-node. This rebuilding
can be done in $O(1)$ rotations on these nodes. The cleaving process ends
when the (2,4)-node containing $x$ is the next such node, with a final
rebuilding which brings $x$ up as root.

We now need to change both subtrees of the root into valid red-black
trees. This will be done in a downwards fashion, one spine at a time. We
describe the process for one spine.

The heights, as (2,4)-trees, of the subtrees hanging from the top of the
spine are increasing when going downwards. There are initially between zero
and three subtrees of each height. For sake of induction, we assume there
are between zero and five subtrees of each height. If there are at least
two subtrees of the next height $h$, one or two (2,4)-nodes of height
$h+1$, hanging from the spine, are formed by rotations on the top nodes of
the spine. These nodes will take part in the next group of (2,4)-nodes of
heights $h+1$, and the process continues with these. Otherwise, there is at
most one subtree of height $h$, and we will look to the next non-empty
group of subtrees (of height $h+k$ for some $k>0$) hanging from the
spine. Of these subtrees, all valid (2,4)-trees, we fold the top-most one
around its root: the binary nodes in its (2,4)-tree root are moved onto the
spine using rotations, making its (2,4)-tree subtrees hang from the
spine. This is repeated $k$ times, which leaves hanging from the spine
between one and three new subtrees of height $h+k-i$, for $i = 1 \dots
k-1$, and between two and four new subtrees of height $h$, at the cost of
$O(k)$ rotations. There are now at least two subtrees of height~$h$ hanging
from the spine, and the process continues with these. By induction, there
will never be more than five trees of any size during this process, and we
will end up with a legal (2,4)-tree, hence a legal red-black tree.

The total time of the split is bounded by the height of the initial tree.

\section*{Part II: Stack with Improved Multipop Operation}
\label{stack}
A stack~\cite{Cormen2001} is a fundamental data structure that asks for the following operations:
\begin{itemize}
\item {P{\scriptsize USH}($x$)} -- insert the element $x$ on the top of the stack.
\item {P{\scriptsize OP}()} -- delete the element from the top of the stack if it is not empty.
\item {M{\scriptsize ULTIPOP}($k$)} -- deletes $k$ elements from the top of the stack, or deletes all elements if the stack has less than $k$ elements.
\end{itemize}
In a pointer machine, a stack is usually implemented as a simple linked list. The operation {P{\scriptsize USH}($x$)} creates an element with value $x$ and insert it as the new head of the list. The operation {P{\scriptsize OP}()} removes the head of the list if this one is not empty. The operation {M{\scriptsize ULTIPOP}($k$)} is performed using $k$ times the {P{\scriptsize OP}()} operation.

Over a sequence of stack operations the amortized cost of each operation is $O(1)$ since an element that has been pushed on the stack can only be popped once from it. Concerning the worst-case performance of the operations it is clear that {P{\scriptsize USH}($x$)} and {P{\scriptsize OP}()} take O(1) wost-case time whereas {M{\scriptsize ULTIPOP}($k$)} takes $O(k)$ worst-case time. Thus the worst-case running time of this operation can reach $O(n)$ if $k\geq n$ where $n$ is the number of elements in the stack. 
Here we develop a stack that improves the running time of {M{\scriptsize ULTIPOP}()}. 

The structure is composed of two parts: first a linked list $L$ of size $\Theta( \log n)$ that contains the most recently pushed elements (as in the original structure) and secondly a red-black tree~\cite{redblack} $T$ containing the remaining elements. Each element $x$ in the stack has a height $h(x)$ which is defined as the number of elements that were in the stack before $x$ was pushed in. The tree $T$ is ordered based on the height of the elements so that the elements in the left (or the right) subtree of an element $x$ have a smaller (or greater) height than $h(x)$. The size $|T|$ of the tree is stored at its root.

The stack operations {P{\scriptsize USH}($x$)} and {P{\scriptsize OP}()} ({M{\scriptsize ULTIPOP}()} is considered later) are essentially performed as in the standard structure. They either add a new head element to the list $L$ or remove it. Thus $L$ is modified and the invariant about its size has to be maintained, i.e., $|L|=\Theta( \log n)$. More specifically we maintain $\log n \leq |L| \leq 5\log n$. In order to do so we use an \emph{extraction} and a \emph{contraction  process}, described below, which transfer $\log n$ elements from the tree into the list or from the list into the tree, respectively, in $O(\log n)$ steps. These processes run incrementally: Whenever an element is popped/pushed from the list, the ongoing extraction/contraction process is advanced by $\Theta(1)$ steps and then the process is suspended until the next operation occurs. Thus the work of an extraction/contraction process is spread over several operations (if not stopped before by a specific kind of {M{\scriptsize ULTIPOP}()} operation, see below) which means that {P{\scriptsize USH}($x$)} and {P{\scriptsize OP}()} still take $O(1)$ worst-case time. The speed of the process is chosen such that the extraction/contraction of $\log n$ elements is completed before $\log n$ push/pop operations have been performed. Whenever $|L|$ reaches $4\log n$ a contraction process is launched and when $|L|$ reaches $2\log n$ an extraction process is launched. Hence there is no more than one extraction/contraction process running at the same time and the size of the list is always maintained between $\log n$ and $5\log n$. 

The operation {M{\scriptsize ULTIPOP}(k)} is performed in the following way: if $k\leq \log n$ then we perform $k$ times the {P{\scriptsize OP}()} operation which takes $O(k)$ worst-case time. Otherwise we finish any pending extraction/contraction process in the stack. We perform $O(1)$ contraction processes until all the elements are contained in the tree $T$. Then we binary search in $T$ for the element $x$ with $h(x)=|T|-k$ and we perform a cut of the tree at the element $x$ using the standard cut operation of red-black trees~\cite{redblack}. The tree containing the $k$ highest elements of the stack is discarded. Finally we perform three extraction processes so that the stack satisfies the invariant. This takes $O(\log n)$ worst-case time in total.

We present the extraction and the contraction process mentioned in the previous description: \\

\noindent {\bf {E{\scriptsize XTRACT}()}} is the operation that transfers the $\log n$ highest elements from the red-black tree into the linked list in $O(\log n)$ time. It first binary searches in $T$ the element $x$ with $h(x)=|T|-\log n$. Once $x$ is found, the tree $T$ is split at $x$ using the standard split operation of red-black trees~\cite{redblack}. The extracted tree containing the highest elements is converted in a linked list in $O(\log n)$ time and finally attached to the end of the list of the stack. \\

\noindent {\bf {C{\scriptsize ONTRACT}()}} is the operation that performs the inverse of {E{\scriptsize XTRACT}()} with the same running time. It transfers the set $S$ of the $\log n$ lowest elements from the linked list into the red-black tree in $O(\log n)$ time. The operation performs a walk along the list to find the position of the first element of $S$, then the list is cut at this precise point. The sublist containing the elements of $S$ is converted into a red-black tree in $O(|S|)=O(\log n)$ time. Finally the newly obtained tree is joined, in $O(\log n)$, with the red-black tree of the stack.\\

\begin{thm}
A stack can be implemented such that the operations {P{\scriptsize USH}()} and {P{\scriptsize OP}()} take $O(1)$ wost-case time and the operation {M{\scriptsize ULTIPOP}($k$)} takes $O(\min\{k,\log n\})$ worst-case time. All these operations take $O(1)$ amortized time. 
\end{thm}

\end{document}